\documentclass[12pt]{article}

\usepackage{a4}
\usepackage{cite}
\usepackage{alg}
\usepackage{booktabs}
\usepackage{amssymb}
\usepackage{amsmath}
\allowdisplaybreaks[1]

\usepackage{listings}
\usepackage{amsthm}

\newtheorem{theorem}{Theorem}
\newtheorem{lemma}[theorem]{Lemma}

\newtheorem{definition}[theorem]{Definition}

\makeatletter
\def\@endtheorem{\endtrivlist}
\makeatother

\newcommand{\Lnum}{\mathrm{Lnum}}
\newcommand{\Rnum}{\mathrm{Rnum}}
\newcommand{\Lnumi}[1]{\mathrm{Lnum}_{#1}}

\newcommand{\Psize}{\mathrm{Psize}}

\newcommand{\kpath}{\textsc{$k$-Path}}
\newcommand{\cutkpath}{\textsc{Cut $k$-Path}}

\newcommand{\ktree}{\textsc{$k$-Tree}}
\newcommand{\kmatching}{\textsc{$r$-Dimensional $k$-Matching}}
\newcommand{\graphmotif}{\textsc{Graph Motif}}
\newcommand{\partialcover}{\textsc{Partial Cover}}

\newcommand{\kstpath}{\textsc{$k$-$(s,t)$-Path}}
\newcommand{\longstpath}{\textsc{Long $(s,t)$-Path}}
\newcommand{\longcycle}{\textsc{Long Cycle}}
\newcommand{\exactdetour}{\textsc{Exact Detour}}

\newcommand{\jl}{j_l}
\newcommand{\jr}{j_r}
\newcommand{\alphal}{\alpha_l}
\newcommand{\alphar}{\alpha_r}

\newcommand{\cla}{c_l}
\newcommand{\cra}{c_r}
\newcommand{\crb}{c'}

\begin{document}

\title{Faster deterministic parameterized algorithm for \kpath}
\author{Dekel Tsur%
\thanks{Department of Computer Science, Ben-Gurion University of the Negev.
Email: \texttt{dekelts@cs.bgu.ac.il}}}
\date{}
\maketitle

\begin{abstract}
In the \kpath\ problem, the input is a directed graph $G$ and an
integer $k\geq 1$, and the goal is to decide whether there is a simple directed
path in $G$ with exactly $k$ vertices.
We give a deterministic algorithm for \kpath\ with time complexity
$O^*(2.554^k)$. This improves the previously best deterministic algorithm
for this problem of Zehavi [ESA 2015] whose time complexity is $O^*(2.597^k)$.
The technique used by our algorithm can also be used to obtain faster
deterministic algorithms for \ktree, \kmatching, \graphmotif, and \partialcover.
\end{abstract}


\section{Introduction}

In the \kpath\ problem, the input is a directed graph $G$ and an
integer $k\geq 1$, and the goal is to decide whether there is a simple directed
path in $G$ with exactly $k$ vertices.
Several papers gave parameterized algorithm for this problem, both
deterministic~\cite{monien1985find,bodlaender1993linear,alon1995color,kneis2006divide,chen2009randomized,fomin2014efficient,shachnai2016representative,zehavi2015mixing} and
randomized~\cite{alon1995color,huffner2008algorithm,kneis2006divide,chen2009randomized,koutis2008faster,williams2009finding,bjorklund2017narrow}.
See Table~\ref{tab:results} for a summary of deterministic parameterized
algorithms for \kpath.
The fastest deterministic parameterized algorithm for \kpath\ was given by
Zehavi~\cite{zehavi2015mixing} and its time complexity is $O^*(2.597^k)$.
In this paper, we give a deterministic algorithm for \kpath\ with time
complexity $O^*(2.554^k)$.

Similarly to the technique of~\cite{zehavi2015mixing}, the technique presented
in this paper can be used to obtain faster deterministic algorithms for other
parameterized problems.
Specifically, for the \ktree, \kmatching, \graphmotif, and \partialcover\
problems, we obtain running times of
$O^*(2.554^k)$, $O^*(2.554^{(r-1)k})$, $O^*(2.554^{2k})$, and $O^*(2.554^k)$,
respectively.
This improves the previously fastest deterministic algorithms for these
problems obtained in~\cite{zehavi2015mixing}, whose running times are
$O^*(2.597^k)$, $O^*(2.597^{(r-1)k})$, $O^*(2.597^{2k})$, and $O^*(2.597^k)$,
respectively.

Our algorithm (as other algorithms for the \kpath\ problem) also solves
a generalization of \kpath\ called \kstpath.
In this problem, the input is
a directed graph $G$, two vertices $s,t$, and an integer $k$,
and the goal is to decide whether there is a simple directed path
from $s$ to $t$ in $G$ with exactly $k$ vertices.
An algorithm for \kstpath\ can be used as a black-box for solving other graph
problems.
Fomin et al.~\cite{fomin2018long} showed that an algorithm for \kstpath\ can be
used to solve the \longstpath\ and \longcycle\ problems.
Bez{\'a}kov{\'a} et al.~\cite{bezakova2017finding} showed that an algorithm for
\kstpath\ can be used to solve the \exactdetour\ problem.
Using our algorithm for \kstpath\ instead the algorithm
of~\cite{zehavi2015mixing} gives faster algorithms for these problems.

Our algorithm is based on the algorithm of Zehavi~\cite{zehavi2015mixing},
with a simple modification: replacing the universal family with an
approximate universal family~\cite{zehavi2016parameterized}.
This causes several straightforward additional changes to the algorithm
of~\cite{zehavi2015mixing} and its analysis.
For completeness, we describe these changes in full.
We note that these changes makes our algorithm substantially simpler than the
algorithm of~\cite{zehavi2015mixing}.
We also note that our algorithm can be extended to solve the weighted variant
of \kpath.
To simplify the presentation, we will describe an algorithm for the unweighted
problem.

\begin{table}
\caption{Deterministic algorithms for the \kpath\ problem.}
\label{tab:results}
\centering
\begin{tabular}{ll}
\toprule
Reference & Running time \\
\midrule
Monien~\cite{monien1985find} & $O^*(k!)$ \\
Alon et al.~\cite{alon1995color} & $O^*(c^k)$ \\
Kneis et al.~\cite{kneis2006divide} & $O^*(16^k)$ \\
Chen et al.~\cite{chen2009randomized} & $O^*(4^{k+o(k)})$ \\
Fomin et al.~\cite{fomin2014efficient} & $O^*(2.851^k)$ \\
Fomin et al.~\cite{fomin2016efficient},
Shachnai and Zehavi~\cite{shachnai2016representative} & $O^*(2.619^k)$ \\
Zehavi~\cite{zehavi2015mixing} & $O^*(2.597^k)$ \\
This paper & $O^*(2.554^k)$\\
\bottomrule
\end{tabular}
\end{table}

\section{Preliminaries}
In this section we describe two tools, representative families and
approximate universal families, that will be used in our algorithm.

Representative families is a tool that is very useful in the
design of parameterized algorithms (cf.~\cite{fomin2016efficient}).
In particular, it was used for giving efficient algorithms for \kpath\
in~\cite{fomin2014efficient,fomin2016efficient,shachnai2016representative}.
\begin{definition}
Let $U$ be a set, called \emph{universe},
and $\mathcal{S}$ be a family of subsets of size $p$ of $U$.
We say that $\widehat{\mathcal{S}} \subseteq \mathcal{S}$ \emph{$q$-represents}
$\mathcal{S}$ if for every set $B\subseteq U$ of size at most $q$,
if there is a set $A\in\mathcal{S}$ disjoint from $B$ then
there is a set $\widehat{A}\in\widehat{\mathcal{S}}$ disjoint from $B$.
\end{definition}

\begin{theorem}[Fomin et al.~\cite{fomin2016efficient},
Shachnai and Zehavi~\cite{shachnai2016representative}]
\label{thm:representative}
There is an algorithm that given $c \geq 1$, integers $p$ and $k \geq p$,
and a family $\mathcal{S}$ of subsets of size $p$ of $U$,
constructs a family $\widehat{\mathcal{S}} \subseteq \mathcal{S}$ that
$(k-p)$-represents $\mathcal{S}$ with size
$\frac{{(ck)}^{k}}{p^p{(ck-p)}^{k-p}} 2^{o(k)} \log|U|$.
The construction time is
$O(|\mathcal{S}|(\frac{ck}{ck-p})^{k-p} 2^{o(k)} \log|U|)$.
%+|\mathcal{S}|\log|\mathcal{S}|)$.
\end{theorem}
Suppose that $|U|=n$ and the size of $\mathcal{S}$ satisfies the bound on the
size of $\widehat{\mathcal{S}}$ of the lemma, namely
$|\mathcal{S}| = O^*(\frac{{(ck)}^{k}}{p^p{(ck-p)}^{k-p}}2^{o(k)})$.
Then, the construction time of $\widehat{\mathcal{S}}$ is
$O^*\left(\frac{{(ck)}^{2k-p}}{p^p{(ck-p)}^{2k-2p}} \cdot 2^{o(k)}\right) =
O^*({\phi_c(p/k)}^k \cdot 2^{o(k)})$,
where
$\phi_c(\alpha) = \frac{c^{2-\alpha}}{\alpha^\alpha{(c-\alpha)}^{2-2\alpha}}$
(we assume that $\phi_c(0)=1$). % and $\phi_1(1) = 1$.

In order to obtain an improved algorithm for \kpath,
Zehavi~\cite{zehavi2015mixing} used the following generalization of
representative families.
\begin{definition}
Let $U_1,\ldots,U_t$ be disjoint sets,
$p_1,\ldots,p_t,q_1,\ldots,q_t$ be non-negative integers,
and $\mathcal{S}$ be a family of subsets of $U=\bigcup_{i\leq t} U_i$ such
that for every $A\in\mathcal{S}$, $|A\cap U_i| = p_i$ for all $i\leq t$.
We say that $\widehat{\mathcal{S}} \subseteq \mathcal{S}$
\emph{$(q_1,\ldots,q_t)$-represents}
$\mathcal{S}$, if for every set $B\subseteq U$ for which
$|B\cap U_i|\leq q_i$ for all $i\leq t$,
if there is a set $A\in\mathcal{S}$ disjoint from $B$ then
there is a set $\widehat{A}\in\widehat{\mathcal{S}}$ disjoint from $B$.
\end{definition}

\begin{theorem}[Zehavi~\cite{zehavi2015mixing}]
\label{thm:generalized-representative}
There is an algorithm that given $c_1,\ldots,c_t \geq 1$, integers
$p_1,\ldots,p_t,k_1,\ldots,k_t$,
and a family $\mathcal{S}$ of subsets of $U=\bigcup_{i\leq t} U_i$ such
that for every $A\in\mathcal{S}$, $|A\cap U_i| = p_i$ for all $i\leq t$,
constructs a family $\widehat{\mathcal{S}} \subseteq \mathcal{S}$ that 
$(k_1-p_1,\ldots,k_t-p_t)$-represents
$\mathcal{S}$ with size $\prod_{i\leq t}(
\frac{{(c_i k_i)}^{k_i}} {{p_i}^{p_i}{(c_i k_i-p_i)}^{k_i-p_i}} \cdot 2^{o(k_i)}
\log|U_i|)$.
The construction time is
$O(|\mathcal{S}|\prod_{i\leq t}(
(\frac{c_i k_i}{c_i k_i-p_i})^{k_i-p_i} 2^{o(k_i)} \log|U_i|))$.
%+|\mathcal{S}|\log|\mathcal{S}|)$.
\end{theorem}
Again, suppose that $|U|=n$ and $\mathcal{S} = O^*(\prod_{i\leq t}(
\frac{{(c_i k_i)}^{k_i}} {{p_i}^{p_i}{(c_i k_i-p_i)}^{k_i-p_i}} \cdot
2^{o(k_i)}))$.
Then, the construction time of $\widehat{\mathcal{S}}$ is
$O^*(\prod_{i\leq t}(
\frac{{(c_i k_i)}^{2k_i-p_i}} {{p_i}^{p_i}{(c_i k_i-p_i)}^{2k_i-2p_i}}
\cdot 2^{o(k_i)})) =
O^*(\prod_{i\leq t} ({\phi_{c_i}(p_i/k_i)}^{k_i} \cdot 2^{o(k_i)}))$.

The algorithm of Zehavi~\cite{zehavi2015mixing} also uses universal families.
\begin{definition}
Let $\mathcal{F}$ be a family of subsets of a set $U$, where $|U|=n$.
We say that $\mathcal{F}$ is an \emph{$(n,p,q)$-universal family} if for every
disjoint sets $A,B\subseteq U$ of sizes $p$ and $q$, respectively,
there is a set $F\in\mathcal{F}$ such that $A\subseteq F$ and
$B\cap F=\emptyset$.
\end{definition}

\begin{lemma}[Fomin et al.~\cite{fomin2014efficient}]\label{lem:universal}
There is an algorithm that given integers $n,p,q$,
constructs an $(n,p,q)$-universal family of size
$O(\binom{p+q}{q} 2^{o(p+q)}\cdot \log n)$ in
$O(\binom{p+q}{q} 2^{o(p+q)}\cdot n\log n)$ time.
\end{lemma}

In order to obtain our improved algorithm, we use a generalization of
universal families called approximate universal
families~\cite{zehavi2016parameterized}.
\begin{definition}
Let $\mathcal{F}$ be a family of subsets of a set $U$, where $|U|=n$.
We say that $\mathcal{F}$ is an \emph{$(n,p,q,\zeta)$-approximate universal
family}
if for every disjoint sets $A,B\subseteq U$ of sizes $p$ and $q$, respectively,
there is a set $F\in\mathcal{F}$ such that
$|A\setminus F| \leq \lfloor\zeta p\rfloor$ and $B\cap F = \emptyset$.
\end{definition}
%We will prove the following lemma in Section~\ref{sec:approximate-universal}.
\begin{lemma}[Zehavi~\cite{zehavi2016parameterized}]\label{lem:approximate-universal}
There is an algorithm that given integers $n,p,q$ and $0 < \zeta < 1$,
constructs an $(n,p,q,\zeta)$-approximate universal family of size
$O(\frac{1}{\eta^p x^{(1-\zeta)p}{(1-x)}^{q+\zeta p}}
\cdot 2^{o(p+q)} \cdot \log n)$ in
$O(\frac{1}{\eta^p x^{(1-\zeta)p}{(1-x)}^{q+\zeta p}} 
\cdot 2^{o(p+q)} \cdot n\log n)$
time, where $x = \frac{(1-\zeta)p}{p+q}$ and
$\eta = \frac{1}{\zeta^\zeta{(1-\zeta)}^{(1-\zeta)}}$.
\end{lemma}
We now define a special type of approximate universal families.
\begin{definition}
An $(n,p,q,\zeta)$-approximate universal family $\mathcal{F}$ is called
\emph{strict} if for every disjoint sets $A,B\subseteq U$ of sizes $p$ and $q$,
respectively,
there is a set $F\in\mathcal{F}$ such that
$|A\setminus F| = \lfloor\zeta p\rfloor$ and $B\cap F = \emptyset$.
\end{definition}

\begin{lemma}\label{lem:strict}
Given an $(n,p,q,\zeta)$-approximate universal family $\mathcal{F}$,
a strict $(n,p,q,\zeta)$-approximate universal family  $\mathcal{F}'$ of size
$O(|\mathcal{F}|n)$ can be constructed in $O(|\mathcal{F}|n^2)$ time.
\end{lemma}
\begin{proof}
Without loss of generality, assume that $U = \{1,\ldots,n\}$.
Define
$\mathcal{F}' = \{ F \cap \{1,\ldots,i\} \colon F\in \mathcal{F}, i \leq n \}$.
It is easy to verify that $\mathcal{F}'$ is a strict
$(n,p,q,\zeta)$-approximate universal family.
\end{proof}

\section{Overview}\label{sec:overview}
In this section we give a high level description of our algorithm.

The \kpath\ problem can be solved in $n^{k+O(1)}$ time by the
following dynamic programming algorithm.
Define $\mathcal{P}^i_v$ to be a family containing all sets $X \subseteq V$
such that $|X|=i$, $v\in X$, and there is a simple path that ends at $v$
whose set of vertices is precisely $X$.
The families $\mathcal{P}^i_v$ are computed using the formula
\[\mathcal{P}^i_v = \bigcup_{u\colon (u,v)\in E}
 \bigcup_{X\in \mathcal{P}^{i-1}_u \colon v \notin X} (X\cup \{v\}). \]
To speed up this algorithm, instead of computing the families
$\mathcal{P}^i_v$,
compute families $\widehat{\mathcal{P}}^i_v \subseteq \mathcal{P}^i_v$
that $(k-i)$-represents $\mathcal{P}^i_v$.
The computation of $\widehat{\mathcal{P}}^i_v$ is done as follows.
First, compute
\[\mathcal{N}^i_v = \bigcup_{u \colon (u,v)\in E}
 \bigcup_{X\in \widehat{\mathcal{P}}^{i-1}_u \colon v \notin X} (X\cup \{v\}).
\]
Then, use Theorem~\ref{thm:representative} to compute a family
$\widehat{\mathcal{P}}^i_v$ that $(k-i)$-represents $\mathcal{N}^i_v$
(note that here the universe is $U=V$).
The time complexity of building $\widehat{\mathcal{P}}^i_v$ is roughly
\[
O^*\left( |\mathcal{N}^i_v| \cdot \left(\frac{ck}{ck-i}\right)^{k-i} \right) =
O^*\left(\frac{{(ck)}^{2k-i}}{i^i{(ck-i)}^{2k-2i}}\right) =
O^*({\phi_c(\alpha)}^k),
\]
where $\alpha = i/k$ (recall that
$\phi_c(\alpha) = \frac{c^{2-\alpha}}{\alpha^\alpha{(c-\alpha)}^{2-2\alpha}}$).
Therefore, the running time of the algorithm is
$O^*(\max_{0 \leq \alpha \leq 1} {\phi_c(\alpha)}^k)$.
The optimal choice for $c$ is $c=1+\frac{1}{\sqrt{5}} \approx 1.45$.
For this choice of $c$, the function $\phi_c(\alpha)$ is maximized when
$\alpha = 1-\frac{1}{\sqrt{5}}\approx 0.55$,
and $\phi_c(1-\frac{1}{\sqrt{5}}) = 3/2+\sqrt{5}/2 \approx
2.619$~\cite{fomin2016efficient,shachnai2016representative}.
Therefore, the running time of the algorithm is $O^*(2.619^k)$.

In order to reduce the time complexity, we use the \emph{color coding}
technique.
Suppose that $G$ contains a path of size $k$,
and let $P=p_1,\ldots,p_k$ be such path.
We call $P$ the \emph{target path}.
We describe an algorithm that is designed to find the specific path $P$
(although it may find a different path of size $k$).
Suppose that we guessed a partition of the vertices of $G$ into sets $L$ and $R$
such that $p_1,\ldots,p_{k/2} \in L$ and $p_{k/2+1},\ldots,p_k \in R$.
We call a partition of the vertices with this property \emph{good}.
Now define the following families.
$\mathcal{P}^i_{L,v}$ is family containing all sets $X \subseteq L$
such that $|X|=i$, $v\in X$, and there is a simple path that ends at $v$
whose set of vertices is precisely $X$.
$\mathcal{P}^i_{R,v}$ is family containing all sets $X \subseteq R$
such that $|X|=i$, $v\in X$, and there is a simple path with $k/2+i$ vertices
that ends at $v$ whose first $k/2$ vertices are in $L$, and the set of the last
$i$ vertices of the path is precisely $X$.
Similarly to before, the algorithm builds families $\mathcal{N}^i_{L,v}$ 
and for each family $\mathcal{N}^i_{L,v}$ it uses
Theorem~\ref{thm:representative} to generate a family
$\widehat{\mathcal{P}}^i_{L,v}$ that $(k/2-i)$-represents $\mathcal{P}^i_{L,v}$.
Similarly, the algorithm compute families $\widehat{\mathcal{P}}^i_{R,v}$
that $(k/2-i)$-represent the families $\mathcal{P}^i_{R,v}$.
The time complexity of computing one representative family is
$O^*(2.619^{k/2}) = O^*(1.619^k)$.
Note that this is considerably faster than the $O^*(2.619^k)$ bound in the first
algorithm.

So far we assumed we guessed a good partition of the vertices into sets $L$ and
$R$.
Since we want a deterministic algorithm, we need to deterministically generate
several partitions such that at least one partition is good.
Since we don't know which partitions are good, the
algorithm performs the dynamic programming stage for every partition.
Therefore, the time complexity is multiplied by the number of partitions.
The generation of a good partition is done using
an $(n,\frac{1}{2} k,\frac{1}{2} k)$-universal family $\mathcal{F}$.
For every $F\in \mathcal{F}$, define sets $L,R$ by taking $L = F$ and
$R = V\setminus F$.
By the definition of universal family, there is at least one good partition.
By Lemma~\ref{lem:universal}, the size of $\mathcal{F}$ is approximately
$\binom{k}{\frac{1}{2} k} \approx 2^k$, which means
that the total time complexity of the algorithm is
$O^*(2^k \cdot 2.619^{k/2}) = O^*(3.2376^k)$, which is worse than the first
algorithm.

The color coding algorithm is less efficient than the first algorithm
since the size of the universal family is too large.
We solve this problem by using a strict
$(n,\frac{1}{2} k,\frac{1}{2} k,\zeta)$-approximate universal family
instead of an $(n,\frac{1}{2} k,\frac{1}{2} k)$-universal family.
Note that the former family is much smaller than the latter
(see Lemma~\ref{lem:universal} and Lemma~\ref{lem:approximate-universal}).
The usage of a strict approximate universal family requires some changes in the
algorithm.
This is because now, the first half of target path contains
$\frac{1}{2}(1-\zeta)k$ vertices from $L$ and $\frac{1}{2}\zeta k$ vertices
from $R$.
Define families
\[ \mathcal{P}^{\jl,\jr}_{L,v} = \{X \in \mathcal{P}^{\jl+\jr}_{L,v}
\colon |X\cap L| = \jl, |X\cap R| = \jr, \}.\]
The algorithm uses Theorem~\ref{thm:generalized-representative}
to generate families $\widehat{\mathcal{P}}^{\jr,\jl}_{L,v}$ that
$(\frac{1}{2}(1-\zeta) k-\jl,\allowbreak \frac{1}{2} (1+\zeta) k-\jr)$-represent
the families $\mathcal{P}^{\jl,\jr}_{L,v}$.
The indices $\jl,\jr$ for which we construct these families are
$0 \leq \jl \leq \frac{1}{2}(1-\zeta) k$ and
$0 \leq \jr \leq \frac{1}{2}\zeta k$.
The time complexity of constructing a single family
$\widehat{\mathcal{P}}^{\jr,\jl}_{L,v}$ is roughly
$O^*(
{\phi_{\cla}(\alphal)}^{\frac{1}{2}(1-\zeta) k}\cdot
{\phi_{\cra}(\alphar)}^{\frac{1}{2}\zeta k}
)$,
where $\alphal = \frac{\jl}{\frac{1}{2}(1-\zeta) k}$ and
$\alphar = \frac{\jr}{\frac{1}{2}(1+\zeta) k}$.
For simplicity of the presentation, suppose that $\zeta = 0.5$ and
$\cla = \cra = 1+1/\sqrt{5}$.
Since $\alphal$ can get values between 0 and 1,
$\phi_{\cla}(\alphal)$ is maximized for
$\alpha_l = 1-\frac{1}{\sqrt{5}}\approx 0.55$ and therefore
$\phi_{\cla}(\alphal) \leq 2.619$ for all $\alphal$.
Note that $\jr \leq \frac{1}{2}\zeta k = \frac{1}{4} k$.
Therefore,  $\alphar$ can get values between 0 and
$\frac{\frac{1}{4} k}{\frac{3}{4} k} = \frac{1}{3}$.
The worst case for $\alphar$ is $\alphar = \frac{1}{3}$, and therefore
$\phi_{\cra}(\alphar) \leq \phi_{\cra}(1/3) \leq 2.313$ for all $\alphar$.
We obtain that the time for generating a representative family is
$O^*(2.619^{k/4} \cdot 2.313^{k/4}) = O^*(1.569^k)$.

In order to improve the running time we use the following idea
from~\cite{zehavi2015mixing}.
Suppose that we guessed a partition $L,R$ that satisfies the following property:
The vertices of the target path $P = p_1,\ldots,p_k$ that are in $L$
are distributed uniformly among $p_1,\ldots,p_{k/2}$.
Namely, the number of vertices among $p_1,\ldots,p_i$ that are in $L$ is
approximately $(1-\zeta) i$ for all $i \leq k/2$.
We call this property the \emph{uniformity property}.
Assuming the uniformity property, we can restrict the pairs of indices
$\jr,\jl$ for which we construct a family
$\widehat{\mathcal{P}}^{\jr,\jl}_{L,v}$ by requiring that
$\frac{\jr}{\jl+\jr} \approx \zeta$.
This prevents the worst case choices of $\alphal$ and $\alphar$
(which are $1-\frac{1}{\sqrt{5}}$ and $\frac{1}{3}$, respectively, when
$\zeta = 0.5$) to occur simultaneously:
If $\alphal = 1-\frac{1}{\sqrt{5}}$ then
$\jl = (1-\frac{1}{\sqrt{5}}) \frac{1}{4}k$
and the requirement $\frac{\jr}{\jl+\jr} \approx \zeta$ implies
(when $\zeta = 0.5$) that $\jr \approx \jl$ and therefore
$\alphar \approx \frac{(1-\frac{1}{\sqrt{5}})\frac{1}{4}k}{\frac{3}{4}k} \approx
0.184$.
Additionally, when $\alphar = \frac{1}{3}$, we have that $\alphal \approx 1$
(since $\jr = \frac{1}{4}k$ and therefore $\jl \approx \frac{1}{4}k$).

Unfortunately, the uniformity property cannot be guaranteed when the partition
$L,R$ is constructed deterministically.
The solution to this problem is as follows:
The first half $p_1,\ldots,p_{k/2}$ of the target path $P=p_1,\ldots,p_k$ is
partitioned into $m$ \emph{sub-paths} each containing $\epsilon k$ internal
vertices (we assume for simplicity that $\epsilon k$ is integer).
%such that the end vertex of the $i$-th path is the start vertex of the
%$(i+1)$-th path.
Now, let $P_1,\ldots,P_m$ be an ordering of the sub-path
such that $|P_i\cap L|\geq |P_{i+1}\cap L|$ for all $i$.
Let $s_i,t_i$ be the first and last vertex of $P_i$, respectively.
Suppose that we guessed the vertices
$s_i$ and $t_i$ for all $i$.
The algorithm works in iterations, where in the $i$-th iteration the algorithm
tries to construct the path $P_i$.
The order property $|P_i\cap L|\geq |P_{i+1}\cap L|$
serves as a replacement for the uniformity property.
Namely, the worst case for the time complexity is when each sub-path $P_i$
contains $(1-\zeta) \epsilon k$ vertices from $L$.
Therefore, the analysis done under the uniformity property also applies here
for large $\epsilon$.

The algorithm works in two stages, where the first stage
tries to construct the sub-paths $P_1,\ldots,P_m$ which are sub-paths
of $p_1,\ldots,p_{k/2}$, and the second stage tries to construct
$P_{m+1}= p_{k/2+1},\ldots,p_k$.
In order to obtain optimal time complexity, the two stages should have
the same time complexities.
To obtain this, we take $P_1,\ldots,P_m$ to be sub-paths of
$p_1,\ldots,p_{\delta k}$ for some constant $\delta$, and
$P_{m+1} = p_{\delta k+1},\ldots,p_k$.

%We conclude this section by describing the algorithm of
%Zehavi~\cite{zehavi2015mixing}.
We note that Zehave~\cite{zehavi2015mixing} used a different approach
for solving the universal family size problem.
The algorithm of~\cite{zehavi2015mixing} guesses a coloring of the vertices
by blue and red
such that exactly $\gamma k$ vertices of the target path are colored blue,
where $\gamma = 0.084$.
Then, only the blue vertices of the graph are partitioned into sets $L$ and $R$.
The size of the universal family used by the algorithm is roughly $2^{\gamma k}$
which is small.
However, since most of the vertices of the target path are colored red,
the improvement in time complexity over the $O^*(2.619^k)$-time algorithm
is small.

\section{The algorithm}
In this section we give a more detailed description of the algorithm and
analyze its time complexity.

Let $\delta,\zeta,\epsilon$ be some constants to be determined later.
To simplify the presentation, we define the following variables:
\begin{align*}
%\Psizem & = k-\frac{1}{\epsilon}-1-(m_l+m_r) \cdot \Psize\\
%m_r & = (\frac{1}{2}-\delta)(\frac{1}{\epsilon}-1)\\
m & = \delta\cdot\frac{1}{\epsilon}\\
\Psize & = \lceil \epsilon k\rceil\\
%\Lnumb & = \lfloor (\frac{1}{2}+\delta) (k-\frac{1}{\epsilon}-1) \rfloor\\
%\Rnumb & = \lceil (\frac{1}{2}-\delta) (k-\frac{1}{\epsilon}-1) \rceil\\
%\Rnump & = \lfloor \zeta \cdot \Lnumb \rfloor\\
\Lnum & = m \cdot \Psize - \lfloor \zeta \cdot m \cdot \Psize \rfloor
= \lceil (1-\zeta)m \cdot \Psize \rceil \\
\Rnum & = k-2-m - \Lnum\\
\Lnumi{i} & = (1-\zeta)i \cdot \Psize
\end{align*}
These variables have the following meanings.
Recall that we partition the target path into $m+1$ sub-paths
$P_1,\ldots,P_{m+1}$
(we can choose $\epsilon$ and $\delta$ such that $m$ is an integer).
$\Psize$ is the number of internal vertices in $P_i$ for all $i \leq m$.
$\Lnum$ and $\Rnum$ are the number of internal vertices of the target path
that are in $L$ and $R$, respectively.
$\Lnumi{i}$ is a lower bound on the number of internal vertices that are in $L$
in $P_1,\ldots,P_i$
(recall that we assume that $|P_i\cap L|\geq |P_{i+1}\cap L|$
for all $i\leq m-1$).

\subsection{Algorithm for {\mdseries\cutkpath}}
\label{sec:cutkpath}

Similarly to Zehavi~\cite{zehavi2015mixing}, we define a problem
called \cutkpath\ (we note that the definition here is different
than the one in~\cite{zehavi2015mixing}).
The input to this problem is a directed graph $G=(V,E)$,
an integer $k$,
a partition of $V$ into disjoint sets $L,R$,
a sequence of distinct vertices $V_e = (v_1,\ldots,v_{m+2})$,
and a permutation $\pi \colon [m] \to [m]$.
We denote
$s_i = v_{\pi(i)}$ and $t_i = v_{\pi(i)+1}$ for all $i \leq m$.
Additionally, $s_{m+1} = v_{m+1}$ and $t_{m+1} = v_{m+2}$.

The goal of the problem is to decide whether there are paths
$P_1,\ldots,P_{m+1}$ with the following properties.
\begin{enumerate}
\item\label{prop:si-ti}
For all $i$, the first vertex of $P_i$ is $s_i$ and the last vertex is $t_i$.
\item\label{prop:disjoint-Ve}
For all $i$, the internal vertices of $P_i$ are disjoint from $V_e$.
\item\label{prop:disjoint-Pj}
For all $i \neq j$, the internal vertices of $P_i$ are disjoint from the
internal vertices of $P_j$.
\item\label{prop:Psize}
For every $i \leq m$, the number of internal vertices of $P_i$ is $\Psize$.
\item\label{prop:Psize-last}
The number of internal vertices of $P_{m+1}$ is $k-2-m-m\cdot \Psize$.
%\item\label{prop:Lnum}
%The number of internal vertices in the paths that are in $L$ is $\Lnum$.
%\item\label{prop:Rnum}
%The number of internal vertices in the paths that are in $R$ is $\Rnum$.
\item\label{prop:Lnumi}
For every $i \leq m$, the number of internal vertices of $P_1,\ldots,P_i$
that are in $L$ is at least $\Lnumi{i}$. % and at most $\Lnum$.
\item\label{prop:Lnum}
The number of internal vertices of $P_1,\ldots,P_{m}$ that are in $L$
is $\Lnum$.
\item\label{prop:disjoint-L}
The internal vertices of $P_{m+1}$ are from $R$.
\end{enumerate}
We note that property~\ref{prop:Lnumi} follows from the assumption that
the paths $P_1,\ldots,P_m$ are ordered such that
$|P_{i}\cap L|\geq |P_{i+1}\cap L|$ for every $i\leq m$.
Also note that properties~\ref{prop:Lnum} and~\ref{prop:disjoint-L} implies that
the number of internal vertices in all the paths that are in $R$
is $\Rnum$.

We now give an algorithm for solving \cutkpath.
We note that the algorithm is based on the algorithm of~\cite{zehavi2015mixing}.
The algorithm consists of two stages.

The first stage constructs the paths $P_1,\ldots,P_m$.
This stage builds a table $M$ in which
$M[i,\jl,\jr,v]$ is a family that $(\Lnum-\jl,\Rnum-\jr)$-represents
(where the universe $U = V$ is partitioned into sets $U_1 = L$ and $U_2 = R$)
the family of all sets of the form
$(P_1 \cup \cdots \cup P_{i-1} \cup P'_i) \setminus V_e$,
where $P_1,\ldots,P_{i-1},P'_i$ are paths such that
\begin{itemize}
\item
$P_1,\ldots,P_{i-1}$ satisfy properties~\ref{prop:si-ti},
\ref{prop:disjoint-Ve}, \ref{prop:disjoint-Pj}, \ref{prop:Psize},
and~\ref{prop:Lnumi}.
\item
$P'_i$ satisfies properties~\ref{prop:disjoint-Ve} and
\ref{prop:disjoint-Pj}.
\item
The first vertex of $P'_i$ is $s_i$ and the last vertex of $P'_i$ is $v$.
\item
The total number of internal vertices of $P_1,\ldots,P_{i-1},P'_i$
that are in $L$ and $R$ is $\jl$ and $\jr$, respectively.
\end{itemize}
The indices $i,\jl,\jr,m$ have the following ranges:
$1 \leq i \leq m$,
$\Lnumi{i-1} \leq \jl \leq \min(i\cdot\Psize,\Lnum)$,
$1+(i-1)\cdot\Psize-\jl \leq \jr \leq i\cdot\Psize-\jl$,
and
\[ v\in \begin{cases}
N^+(s_i) \setminus V_e & \text{if } \jl+\jr = 1+(i-1)\cdot\Psize\\
N^-(t_i) \setminus V_e & \text{if } \jl+\jr = i\cdot\Psize\\
V\setminus V_e & \text{otherwise}
\end{cases}
\]
where $N^+(x)$ and $N^-(x)$ are the sets of out-neighbors and in-neighbors of
$x$, respectively.
If at least one of $i,\jl,\jr,v$ does not satisfy the requirements above,
assume that $M[i,\jl,\jr,v] = \emptyset$.
Note that the bounds on $\jr$ can be rewritten as
$1+(i-1)\cdot\Psize \leq \jl+\jr \leq i\cdot\Psize$.
Since the number of internal vertices in $P_1,\ldots,P_{i-1}$ is
$(i-1)\cdot\Psize$ (due to property~\ref{prop:Psize}),
this inequality forces the number of internal vertices of
$P'_i$ to be between 1 and $\Psize$.

The computation of an entry $M[i,\jl,\jr,v]$ is done as follows.
If $\jl+\jr > 1+(i-1)\cdot\Psize$ then
\[
M[i,\jl,\jr,v] = \begin{cases}
\{A\cup\{v\}\colon A\in\bigcup_{u \in N^-(v) \setminus V_e} M[i,\jl-1,\jr,u]\}
& \text{if }v \in L\\
\{A\cup\{v\}\colon A\in\bigcup_{u \in N^-(v) \setminus V_e} M[i,\jl,\jr-1,u]\}
& \text{otherwise}
\end{cases}
\]
If $\jl+\jr = 1+(i-1)\cdot\Psize$ and $i > 1$ then
\[
M[i,\jl,\jr,v] = \begin{cases}
\{A\cup\{v\}\colon A\in\bigcup_{u \in N^-(t_{i-1}) \setminus V_e}
 M[i-1,\jl-1,\jr,u]\}
& \text{if }v \in L\\
\{A\cup\{v\}\colon A\in\bigcup_{u \in N^-(t_{i-1}) \setminus V_e}
 M[i-1,\jl,\jr-1,u]\}
& \text{otherwise}
\end{cases}
\]
Finally, if $\jl+\jr = 1$ then
\[
M[i,\jl,\jr,v] = \begin{cases}
\{\{v\}\}
& \text{if ($v \in L$ and $\jl = 1$) or ($v \in R$ and $\jr = 1$)}\\
\emptyset
& \text{otherwise}
\end{cases}
\]
Then, use Theorem~\ref{thm:generalized-representative}
to find a family that
$(\Lnum-\jl,\Rnum-\jr)$-represents $M[i,\jl,\jr,v]$, and replace
$M[i,\jl,\jr,v]$ with this family.
Theorem~\ref{thm:generalized-representative} is applied with
$U_1 = L$ and $U_2 = R$, and with constants $\cla,\cra$.

The second stage of the algorithm constructs the path $P_{m+1}$.
This stage constructs a table $K[j,v]$ in which
$K[j,v]$ is a family that $(\Rnum-j)$-represents
the family of all sets of the form
$(P_1 \cup \cdots \cup P_m \cup P'_{m+1}) \setminus (V_e \cup L)$,
where $P_1,\ldots,P_m,P'_{m+1}$ are paths such that
\begin{itemize}
\item
$P_1,\ldots,P_m$ satisfy properties~\ref{prop:si-ti}, \ref{prop:disjoint-Ve},
\ref{prop:disjoint-Pj}, \ref{prop:Psize}, \ref{prop:Lnumi}, \ref{prop:Lnum},
and~\ref{prop:disjoint-L}.
\item
$P'_{m+1}$ satisfies properties~\ref{prop:disjoint-Ve},
\ref{prop:disjoint-Pj}, and~\ref{prop:disjoint-L}.
\item
The first vertex of $P'_{m+1}$ is $s_{m+1}$ and the last vertex of $P'_{m+1}$
is $v$.
\item
The number of internal vertices of $P_1,\ldots,P_m,P'_{m+1}$
that are in $R$ is $j$.
\end{itemize}
The indices $j,v$ have the following ranges:
$1 + m \cdot \Psize -\Lnum \leq j \leq \Rnum$ and
\[ v\in \begin{cases}
R \cap (N^+(s_{m+1}) \setminus V_e) & \text{if } j = 1+m\cdot\Psize-\Lnum\\
R \cap (N^-(t_{m+1}) \setminus V_e) & \text{if } j = \Rnum\\
R\setminus V_e & \text{otherwise}
\end{cases}
\]
The computation of an entry $K[j,v]$ is done as follows.
First, perform
\[
K[j,v] =
\begin{cases}
\{A\cup\{v\}\colon A\in \bigcup_{u \in R \cap (N^-(v)\setminus V_e)} K[j-1,u]\}
& \text{if } j > 1+m\cdot\Psize-\Lnum\\
\{A\cup\{v\}\colon A\in \bigcup_{u \in N^-(t_m) \setminus V_e}
 M[m,\Lnum,j-1,u]\}
& \text{otherwise}
\end{cases}
\]
Then, use Theorem~\ref{thm:representative}
to find a family that $(\Rnum-j)$-represents $K[j,v]$,
and replace $K[j,v]$ with this family.
Theorem~\ref{thm:representative} is applied with $U = R$ and constant $\crb$.

After the second stage, if there is a vertex $v$ such that
$K[\Rnum,v] \neq \emptyset$, the algorithm returns `yes'.
Otherwise, the algorithm returns `no'.

\subsection{Algorithm for {\mdseries\kpath}}
\label{sec:kpath}

The following algorithm solves the \kpath\ problem, using the algorithm
for \cutkpath\ of the previous section.

\begin{algtab}
Construct a strict $(n,\Lnum,\Rnum,\zeta)$-approximate universal family
$\mathcal{F}$ over the universe $V$.\\
\algforeach{sequence of distinct vertices $V_e = (v_1,\ldots,v_{m+2})$}
  \algforeach{permutation $\pi \colon [m] \to [m]$}
    \algforeach{$F \in \mathcal{F}$}
      $L \gets F$ and $R \gets V\setminus F$.\\
      Run the \cutkpath\ algorithm on the instance $(G,k,L,R,V_e,\pi)$.\\
      \algifthen{the algorithm returned `yes'}{\algreturn `yes'.}
    \algend
  \algend
\algend
\algreturn `no'\\
\end{algtab}

\subsection{Analysis}
We now analyze the time complexity of our algorithm.
Consider the algorithm for \cutkpath\ of Section~\ref{sec:cutkpath}.
By Theorem~\ref{thm:generalized-representative}, the time complexity of
the first stage of the algorithm is $O^*(X_1 2^{o(k)})$, where
\begin{multline*}
X_1 =
 \max_{i=1}^m
 \max_{\jl = \Lnumi{i-1}}^{\Lnum}
 \max_{\jr = 1+(i-1)\cdot\Psize-\jl}^{i\cdot\Psize-\jl}\\
\frac{{(\cla\cdot\Lnum)}^{2\cdot\Lnum-\jl}}
     {\jl^{\jl}{(\cla\cdot\Lnum-\jl)}^{2\cdot\Lnum-2\jl}}
\cdot
\frac{{(\cra\cdot\Rnum)}^{2\cdot\Rnum-\jr}}
     {\jr^{\jr}{(\cra\cdot\Rnum-\jr)}^{2\cdot\Rnum-2\jr}}.
\end{multline*}
Our goal is to estimate $Y_1 = X_1^{1/k}$.
To simplify the analysis, we redefine the values of the following variables:
\begin{align*}
\Psize & = \epsilon k\\
\Lnum & = (1-\zeta) \delta k\\
\Rnum & = (1-\delta+\zeta\delta) k\\
\Lnumi{i} & = (1-\zeta)(i+1)\cdot \Psize
%            = (i+1) \epsilon (1-\zeta) k\\
\end{align*}
Note that since we can assume that $k$ is large enough and that $\epsilon$ is
small enough, the value of $Y_1$ for the new definitions of the variables
is arbitrarily close to the value of $Y_1$ for the old definitions.
Define $\alphal = \frac{\jl}{\Lnum}$ and $\alphar = \frac{\jr}{\Rnum}$.
The range of $\jl$ in the definition of $X_1$ (over all $i$) is
$0 \leq \jl \leq \Lnum$.
Therefore, $0 \leq \alphal \leq 1$.
The range of $\jl$ in the second maximum in the definition of $X_1$ implies that
$\jl \geq \Lnumi{i-1} = i(1-\zeta)\epsilon k = i\epsilon \cdot\Lnum/\delta$.
Therefore, $i \leq \frac{1}{\epsilon}\delta \cdot \frac{\jl}{\Lnum}
= \frac{1}{\epsilon}\delta\alphal$.
The range of $\jr$ in the third maximum implies that
\begin{align*}
\jr & \leq i\cdot\Psize - \jl
\leq \frac{1}{\epsilon}\delta\alphal \cdot \epsilon k
  - \Lnum\cdot\alphal \\
& = \delta\alphal k - (1-\zeta)\delta k \cdot \alphal
 =  \zeta\delta\alphal k.
\end{align*}
Therefore,
\[ \alphar = \frac{\jr}{\Rnum} \leq
\frac{\zeta\delta}{1-\delta+\zeta\delta}
\cdot \alphal.
\]
We obtain that
\begin{align*}
Y_1 & =
\max_{0 \leq \alphal \leq 1}
\max_{0 \leq \alphar \leq \frac{\zeta\delta}{1-\delta+\zeta\delta}\cdot \alphal}
{\phi_{\cla}(\alphal)}^{\Lnum/k}\cdot
{\phi_{\cra}(\alphar)}^{\Rnum/k}
\\ & =
\max_{0 \leq \alphal \leq 1}
\max_{0 \leq \alphar \leq \frac{\zeta\delta}{1-\delta+\zeta\delta}\cdot \alphal}
{\phi_{\cla}(\alphal)}^{(1-\zeta)\delta}\cdot
{\phi_{\cra}(\alphar)}^{1-\delta+\zeta\delta}.
\end{align*}
The time complexity of the second stage is $O^*(X_2 2^{o(k)})$, where
\[
X_2 = \max_{j = 1 + m \cdot \Psize-\Lnum}^{\Rnum}
\frac{{(\crb\cdot\Rnum)}^{2\cdot\Rnum-j}}
     {j^{j}{(\crb\cdot\Rnum-j)}^{2\cdot\Rnum-2j}}.
\]
Under the simplified definitions of the variables,
we have that $j$ satisfies
$j \geq m \cdot \Psize - \Lnum =
\zeta\delta k$.
Therefore for $\alpha = \frac{j}{\Rnum}$ we have
$\alpha \geq \frac{\zeta\delta k}{\Rnum} =
\frac{\zeta\delta}{1-\delta+\zeta\delta}$.
Let $Y_2 = X_2^{1/k}$.
We obtain that
\[ Y_2 =
\max_{\frac{\zeta\delta}{1-\delta+\zeta\delta} \leq \alpha \leq 1}
{\phi_{\crb}(\alpha)}^{1-\delta+\zeta\delta}.
\]
The time complexity of the algorithm for \cutkpath\ is
$O^*({\max(Y_1,Y_2)}^k 2^{o(k)})$.
The algorithm for \kpath\ generates $O(|\mathcal{F}|\cdot n^{O(1/\epsilon)})$
instances of \cutkpath,
where $\mathcal{F}$ is a strict $(n,\Lnum,\Rnum,\zeta)$-approximate universal
family.
By Lemma~\ref{lem:approximate-universal} and Lemma~\ref{lem:strict}
%to construct $\mathcal{F}$ with $x = (1-\zeta)(\frac{1}{2}+\delta)$,
we have that
$|\mathcal{F}| = O^*(Y_3^k 2^{o(k)})$ where
\[Y_3 = \frac{
{(\zeta^\zeta {(1-\zeta)}^{(1-\zeta)})}^{\delta} }{
{((1-\zeta)\delta)}^{(1-\zeta)\delta}
{(1-\delta+\zeta\delta)}^{1-\delta+\zeta\delta} }.
\]
Therefore, the time complexity of the algorithm for \kpath\ is
$O^*({\max(Y_1,Y_2)}^k \cdot Y_3^k 2^{o(k)})$.
We now choose the following parameters in order to minimize the time complexity:
$\epsilon = 10^{-10}$,
$\delta = 0.49533$, $\zeta = 0.712$, $\cla = 1.136$,
$\cra = 1.645$, and $\crb = 1+\frac{1}{\sqrt{5}} \approx 1.447$.
Under this choice of parameters, $\max(Y_1,Y_2) \cdot Y_3 < 2.5537$.
The value of $Y_1$ is maximized when $\alphal \approx 0.864$ and
$\alphar \approx 0.356$.
The value of $Y_2$ is maximized when $\alpha =  1-\frac{1}{\sqrt{5}} \approx
0.553$.
The values of these parameters were obtained with a Python script.
See the appendix for details.

We note that there is a special case which was omitted in the analysis above:
In the computation of $K[j,v]$ for $j = 1+m\cdot\Psize-\Lnum$, the algorithm
first builds a family $K[j,v]$ of size
\begin{align*}
O^*\left(\bigcup_{u\in V} M[m,\Rnum,j-1,u]\right) &=
O^*\left( \frac{{(\cra\cdot\Rnum)}^{\Rnum}}
     {j^{j}{(\cra\cdot\Rnum-j)}^{\Rnum-j}} \right)\\
 & =
O^*\left( \left(\frac{\cra}{\alpha^\alpha(\cra-\alpha)^{1-\alpha}}
          \right)^{(1-\delta+\zeta\delta)k} % \Rnum/k}
\right)
\end{align*}
where $\alpha = j/\Rnum \approx \frac{\zeta\delta}{1-\delta+\zeta\delta}$.
Then, the time for constructing the representative family is
\[
O^*\left( |K[j,v]| \cdot \frac{{(\crb\cdot\Rnum)}^{\Rnum-j}}
     {{(\crb\cdot\Rnum-j)}^{\Rnum-j}} \right) =
O^*\left( \left(\frac{\cra\cdot\crb^{1-\alpha}}
{\alpha^\alpha(\cra-\alpha)^{1-\alpha}(\crb-\alpha)^{1-\alpha}}
          \right)^{(1-\delta+\zeta\delta)k}
\right)
\]
This does not change the time complexity of the algorithm since for choice
of the parameters given above,
the last expression is smaller than
$Y_2^k$.

\bibliographystyle{plain}
\bibliography{parameterized}

\appendix

\section{Computation of optimal parameters}
In this section we give a Python script for finding optimal parameters
for the algorithm.
To speed-up the computation, we use the following observation.
Consider the computation of
\[ Y_1 = \max_{0 \leq \alphal \leq 1}
\max_{0 \leq \alphar \leq \frac{\zeta\delta}{1-\delta+\zeta\delta}\cdot \alphal}
{\phi_{\cla}(\alphal)}^{(1-\zeta)\delta}\cdot
{\phi_{\cra}(\alphar)}^{1-\delta+\zeta\delta}.\]
Let $\alphar^*$ be the value of $\alpha \in [0,1]$ that maximizes
$\phi_{\cra}(\alpha)$.
The function $\phi_{\cra}(\alphar)$ is monotonically increasing in the range
$[0,\alphar^*]$.
Moreover, for the relevant values of $\delta$ and $\zeta$ we have that
$\frac{\zeta\delta}{1-\delta+\zeta\delta}\cdot \alphal \leq
\frac{\zeta\delta}{1-\delta+\zeta\delta} < 0.5 < \alphar^*$.
Therefore,
\[ Y_1 = \max_{0 \leq \alphal \leq 1}
{\phi_{\cla}(\alphal)}^{(1-\zeta)\delta}\cdot
{\phi_{\cra}\left(
\frac{\zeta\delta}{1-\delta+\zeta\delta}\cdot\alphal
\right)}^{1-\delta+\zeta\delta}.\]
The computation of $Y_1$ using the second formula is much faster than using the first formula.
Similarly, we have that
$\frac{\zeta\delta}{1-\delta+\zeta\delta} < 0.5 < \alpha^*$,
where $\alpha^*$ is the value of $\alpha \in [0,1]$ that maximizes
$\phi_{\crb}(\alpha)$.
Therefore,
\[ Y_2 =
{\phi_{\crb}(\alpha^*)}^{1-\delta+\zeta\delta}.
\]
It follows that the optimal value for $\crb$ is $\crb = 1+\frac{1}{\sqrt{5}}$
and $Y_2 = {\phi_{\crb}(1-\frac{1}{\sqrt{5}})}^{1-\delta+\zeta\delta} = 
(3/2+\sqrt{5}/2)^{1-\delta+\zeta\delta}$.

The script for computing the value of the parameters is as follows.
%\lstinputlisting[language=Python,firstline=2]{kpath.py}
\begin{lstlisting}[language=Python]
from math import sqrt

def frange(a,b, steps):
	return [a+(b-a)*float(x)/steps for x in range(steps+1)]

def phi(alpha,c):
	return c**(2-alpha)/alpha**alpha/(c-alpha)**(2-2*alpha)

def calc_Y1(delta,zeta,cl,cr):
	Lnum = (1-zeta)*delta
	Rnum = 1-delta+zeta*delta
	Y1 = 0
	for alphal in frange(0, 1.0, 1000):
		alphar = zeta*delta/(1-delta+zeta*delta)*alphal
		Y1 = max(Y1, phi(alphal,cl)**Lnum * phi(alphar,cr)**Rnum)
	return Y1

def calc_Y2(delta,zeta):
	Rnum = 1-delta+zeta*delta
	return (1.5+sqrt(5)/2)**Rnum

def calc_Y3(delta,zeta):
	x = (1-zeta)*delta
	return (zeta**zeta*(1-zeta)**(1-zeta))**delta/x**x/(1-x)**(1-x)

zeta_range = frange(0.10, 0.9, 80)
cl_range = frange(1.00, 1.40, 40)
cr_range = frange(1.40, 1.80, 40)
best_Y = 1e10
for cl in cl_range:
	for cr in cr_range:
		for zeta in zeta_range:
			delta_a = 0.25
			delta_b = 0.75
			for i in range(30):
				delta = (delta_a+delta_b)/2
				Y1 = calc_Y1(delta,zeta,cl,cr)
				Y2 = calc_Y2(delta,zeta)
				if Y1 > Y2:
					delta_b = delta
				else:
					delta_a = delta
			Y = Y1*calc_Y3(delta,zeta)
			if Y < best_Y:
				best_Y = Y
				best_params = [cl,cr,zeta,delta]
print best_Y
print best_params
\end{lstlisting}

\end{document}